\documentclass[twoside,leqno,twocolumn]{article}
\usepackage[letterpaper]{geometry}
\usepackage{ltexpprt}

\usepackage[english]{babel}
\usepackage[T1]{fontenc}
\usepackage[utf8]{inputenc}
\usepackage{graphicx}
\usepackage{subcaption}
\usepackage{caption}
\usepackage{booktabs}
\usepackage{url}
\usepackage{balance}

\usepackage{enumitem}
\newlist{dsitemize}{itemize}{1}
\setlist[dsitemize,1]{label=,leftmargin=0mm}

\usepackage{graphicx}
\usepackage{todonotes}
\usepackage{multirow}
\usepackage{tikz}
\usepackage{skull}
\usepackage{amsmath}
\usepackage{amssymb}
\usepackage{amsfonts}
\usepackage{thmtools}		
\usepackage{mleftright}
\usepackage{thm-restate}
\usepackage[mathic=true]{mathtools}
\usepackage{fixmath} 
\usepackage{siunitx}
\usepackage{pifont}
\usepackage{nicefrac}
\newcommand{\cmark}{\ding{51}}%
\newcommand{\xmark}{\ding{55}}%

\usepackage{enumitem}
\setlist[enumerate]{noitemsep, topsep=0.5\topsep}
\setlist[description]{noitemsep, topsep=0.5\topsep}
\setlist[itemize]{noitemsep, topsep=0.5\topsep}
\usepackage{ellipsis}
\usepackage{xspace}
\usepackage{algorithm}
\usepackage{algorithmicx}
\usepackage{algpseudocode}
\algdef{SE}{ParForAllLoop}{EndParForAllLoop}[1]{\textbf{parallel for}\xspace\(\mbox{#1}\) \textbf{do}}{\textbf{end}}
\usepackage{ifthen}
\newcommand{\CC}[1][]{$\text{C\hspace{-.25ex}}^{_{_{_{++}}}}
	\ifthenelse{\equal{#1}{}}{}{\text{\hspace{-.625ex}#1}}$}

\DeclarePairedDelimiter{\norm}{\lVert}{\rVert}
\newcommand{\tg}{\mathbf{G}}
\newcommand{\tge}{\mathbf{E}}

\usepackage[
pdfa,
hidelinks,
pdftex, 
pdfdisplaydoctitle,
pdfpagelabels,
pdftitle={Temporal Graph Kernels for Classifying Dissemination Processes},
pdfsubject={},
pdfkeywords={},
pdfproducer={Latex with the hyperref package},
pdfcreator={pdflatex}
]{hyperref}
\makeatletter
\def\thmt@refnamewithcomma #1#2#3,#4,#5\@nil{%
	\@xa\def\csname\thmt@envname #1utorefname\endcsname{#3}%
	\ifcsname #2refname\endcsname
	\csname #2refname\expandafter\endcsname\expandafter{\thmt@envname}{#3}{#4}%
	\fi
}
\makeatother
\usepackage[capitalise,noabbrev]{cleveref}   

\newcommand{\new}[1]{\emph{#1}}
\newcommand{\sz}{\scriptstyle }

\definecolor{LOGreen}{rgb}{0.0, 0.6, 0.0}

\usepackage{tikz}
\usetikzlibrary{arrows,decorations.pathmorphing}
\tikzset{
	vertex/.style={circle,draw,inner sep=01pt,minimum size=1.3em},
	vertex2/.style={circle,draw,inner sep=01pt,minimum size=1.5em},	
	vertexa/.style={circle,draw,fill=red, minimum size=2mm,inner sep=1pt, font=\scriptsize},	
	vertexb/.style={circle,draw,fill=black, minimum size=2mm,inner sep=1pt, font=\scriptsize},	
	vertexc/.style={circle,draw,fill=green, minimum size=2mm,inner sep=1pt, font=\scriptsize},	
	vertexd/.style={circle,draw,fill=black, minimum size=2mm,inner sep=0pt},
	vertexldg/.style={circle,draw,fill=gray, minimum size=2mm,inner sep=0pt,color=gray},		
	vertexld/.style={draw,fill=white, minimum size=4.1mm,inner sep=0pt},	
	ldgedge/.style={->,> = latex', font=\footnotesize,color=gray,dashed},
	dedge/.style={->,> = latex', font=\footnotesize},
	wedge/.style={->,> = latex', font=\footnotesize, dashed},
	edge/.style={-,> = latex', font=\footnotesize},
}
\newtheorem{definition}{Definition}

\providecommand{\keywords}[1]
{
    \small	
    \textbf{{Keywords---}} #1
}

\begin{document}
	\title{Temporal Graph Kernels for Classifying Dissemination Processes}
	\date{}
	
	\author{Lutz Oettershagen\thanks{Institute of Computer Science I, University of Bonn, 
        \texttt{\{lutz.oettershagen, petra.mutzel\}@cs.uni-bonn.de}.
        }\and  
		Nils M.~Kriege\thanks{Department of Computer Science, TU Dortmund University, \texttt{\{nils.kriege, christopher.morris\}@tu-dortmund.de}.}\and 
		Christopher Morris\footnotemark[2]\and 
		Petra Mutzel\footnotemark[1]}
	\maketitle

\begin{abstract}\small \baselineskip=9pt
    Many real-world graphs or networks are temporal, e.g., in a social network persons only interact at specific points in time. This information directs \emph{dissemination processes} on the network, 
    such as the spread of rumors, fake news, or diseases. However, the current state-of-the-art methods for supervised graph classification are designed mainly for static graphs and may not be able to capture temporal information. Hence, they are not powerful enough to distinguish between graphs modeling different dissemination processes. To address this, we introduce a framework to lift standard graph kernels to the temporal domain. Specifically, we explore three different approaches and investigate the trade-offs between loss of temporal information and efficiency. Moreover, to handle large-scale graphs, we propose stochastic variants of our kernels with provable approximation guarantees. We evaluate our methods on a wide range of real-world social networks. Our methods beat static kernels by a large margin in terms of accuracy while still being scalable to large graphs and data sets. Hence, we confirm that taking temporal information into account is crucial for the successful classification of dissemination processes.
    \\\keywords{Temporal graphs, Classification, Kernels}
\end{abstract}

\section{Introduction}
Linked data arise in various domains, e.g., in chem- and bioinformatics, social network analysis and computer vision, and can be naturally represented as graphs. Therefore, machine learning with graphs has become an active research area of increasing importance.
A prominent method primarily used for supervised graph classification with support vector machines are \new{graph kernels}, which compute a similarity score between pairs of graphs. In the last fifteen years, a plethora of graph kernels has been published~\cite{Kriege2019}. With few exceptions, these are designed for static graphs and cannot utilize temporal information. 
However, real-world graphs often have temporal information attached to the edges, e.g., modeling times of interaction in a social network, which directs any dissemination process, i.e., the spread of information over time, in the temporal graph.
Consider, for example, a social network where persons A and B were in contact before persons B and C became in contact. Consequently, information may have been passed from persons A to C but not vice versa.
Hence, whenever such implicit direction is essential for the learning task, static graph kernels will inevitably fail. %

To further exemplify this, assume that a \mbox{(sub-)group} in a  social network suffers from unspecific symptoms, that occurred at some point in time and are probably caused by a disease. Here, nodes represent persons, and the edges represent contacts between them at certain points in time, cf.~Figure~\ref{fig:example_spread}.
An obvious question now is whether an infectious disease causes the symptoms. However, dissemination processes are typically complex, since persons may have different risk factors of becoming infected, the transmission rate is unknown and, finally, the network structure itself might suffer from noise. Therefore, this question is difficult to answer by just analyzing a single network. But, assume
that similar network data of past epidemics is available. Hence, we can model the detection of dissemination process of a disease as a supervised graph classification problem. In the simplest case, one class contains graphs under a dissemination process of a disease and the other class consists of graphs, where the node labels cannot be explained by the temporal information. 
Furthermore, notice that infections, or disseminated information in general, often may not be recognized or not reported~\cite{who2019rep43}. 
Therefore, we additionally consider the scenario where disseminated information is incomplete. 
Finally, observe that the above learning problem can also model the detection of other dissemination processes in networks, e.g., dissemination of fake news in social media~\cite{Vos+2018} or viruses in computer or mobile phone networks~\cite{Adams2016}. 

The key to solving these classification tasks are methods that adequately take the temporal characteristics of such graphs into account. 
We consider \new{temporal graphs}, where edges exist at specific points in time 
and node labels, e.g., representing infected and non-infected, may change over time. 

\begin{figure}
   \begin{subfigure}[t]{0.25\columnwidth}
        \begin{tikzpicture}[scale=0.8]
        \node[vertexd] (1) at (-0.5,3) {};
        \node[vertexd] (2) at (0.5,3) {};
        \node[vertexd] (3) at (-0.5,2) {};
        \node[vertexd] (4) at (0.5,2) {};
        \node[vertexd] (5) at (-0.5,1) {};
        \node[vertexd] (6) at (0.5,1) {};
        \path[->] 
        (1)  edge[edge] node[midway,above] {\scriptsize$3$} (2)
        (1)  edge[edge] node[midway,left] {\scriptsize$2$} (3)
        (1)  edge[edge] node[midway,below,sloped] {\hspace{6pt}\scriptsize$1,\!2$} (4)
        (2)  edge[edge] node[midway,right] {\scriptsize$3,\!4$} (4)
        (3)  edge[edge] node[midway,below,sloped] {\scriptsize$1$} (6)
        (3)  edge[edge] node[midway,left] {\scriptsize$1$} (5)
        (4)  edge[edge] node[midway,right] {\scriptsize$4$} (6)
        (5)  edge[edge] node[midway,below] {\scriptsize$3$} (6);    
        
        \end{tikzpicture}
        \subcaption{}
        \label{fig:example_spread_a}
    \end{subfigure}\hfill
    \begin{subfigure}[t]{0.8\columnwidth}
        \begin{tikzpicture}[scale=.71]
        \node at (2.5,0.5) {\scriptsize Day 1};
        \node at (5,0.5) {\scriptsize Day 2};
        \node at (7.5,0.5) {\scriptsize Day 3};
        \node at (10,0.5) {\scriptsize Day 4};
        
        \node[vertexa] (1a) at (2,3) {};
        \node[vertexd] (2a) at (3,3) {};
        \node[vertexd] (3a) at (2,2) {};
        \node[vertexb] (4a) at (3,2) {};
        \node[vertexd] (5a) at (2,1) {};
        \node[vertexd] (6a) at (3,1) {};
        \path[->] 
        (1a)  edge[edge] (4a) %
        (3a)  edge[edge]   (6a)
        (3a)  edge[edge] (5a)
        ;    
        
        \node[vertexa] (1b) at (4.5,3) {};
        \node[vertexd] (2b) at (5.5,3) {};
        \node[vertexb] (3b) at (4.5,2) {};
        \node[vertexa] (4b) at (5.5,2) {};
        \node[vertexd] (5b) at (4.5,1) {};
        \node[vertexd] (6b) at (5.5,1) {};
        \path[->] 
        (1b)  edge[edge]  (3b)%
        (1b)  edge[edge]   (4b)
        ;    
        
        \node[vertexa] (1c) at (7,3) {};
        \node[vertexb] (2c) at (8,3) {};
        \node[vertexa] (3c) at (7,2) {};
        \node[vertexa] (4c) at (8,2) {};
        \node[vertexd] (5c) at (7,1) {};
        \node[vertexd] (6c) at (8,1) {};
        \path[->] 
        (1c)  edge[edge]  (2c)%
        (2c)  edge[edge]  (4c)%
        (5c)  edge[edge]  (6c)
        ;    
        
        \node[vertexa] (1d) at (9.5,3) {};
        \node[vertexa] (2d) at (10.5,3) {};
        \node[vertexa] (3d) at (9.5,2) {};
        \node[vertexa] (4d) at (10.5,2) {};
        \node[vertexd] (5d) at (9.5,1) {};
        \node[vertexb] (6d) at (10.5,1) {};
        \path[->] 
        (2d)  edge[edge]  (4d)
        (4d)  edge[edge]  (6d)%
        ;        
        \end{tikzpicture}
        \subcaption{}\vspace{-15pt}
        \label{fig:example_spread_b}
    \end{subfigure}
    \caption{Example of an epidemic outbreak:
        (a) A temporal graph where the vertices represent persons and edges their interaction.
        (b) For each day $t$ there is a static graph containing only the edges existing on day $t$. 
        Infected vertices are labeled red.}\vspace{-20pt}
    \label{fig:example_spread}
\end{figure}
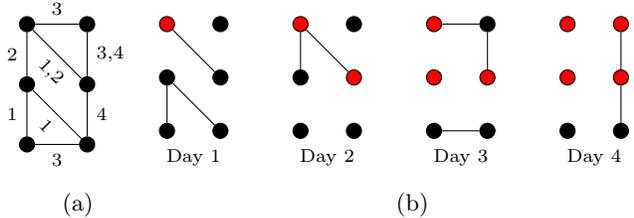
\subsection{Contributions}
We propose graph kernels for classifying dissemination processes in temporal graphs.
More specifically, our main contributions are:
\begin{enumerate}%
    \item We introduce three different techniques to map temporal graphs to static graphs such that conventional, static kernels can be successfully applied. For all three approaches, we analyze the trade-offs between loss of temporal information and size of the transformed graph.
    \item For large-scale problems, we present a stochastic kernel directly based on temporal walks with provable approximation guarantees.
    \item We comprehensively evaluate our methods using real-world data sets with simulations of epidemic spread.
    Our results confirm that taking temporal information into account is crucial for the detection of dissemination processes. %
\end{enumerate}
\subsection{Related Work}
Graph kernels have been studied extensively in the past 15 years, see~\cite{Kriege2019}. Almost all kernels focus on static graphs.
Important approaches include random-walk and shortest paths based kernels~\cite{Gaertner2003,Sugiyama2015,Bor+2005,Kri+2017b}, as well as the Weisfeiler-Lehman subtree kernel~\cite{She+2011,Mor+2017}. 
Further recent works focus on assignment-based approaches~\cite{Kri+2016,Nik+2017}, spectral approaches~\cite{Kon+2016}, and graph decomposition approaches~\cite{Nik+2018}.

There has been some work considering dynamic graphs.
In~\cite{Li+2015} a family of algorithms to recompute the random walk kernel efficiently when graphs are modified is presented.
Wu et al.~\cite{Wu/etal/2014a} propose graph kernels for human action recognition in video sequences. To this end, they encode the features of each frame as well as the dynamic changes between successive frames by separate graphs, which are then compared by random walk kernels.
Paaßen et al.~\cite{Paassen/etal/2017a} use graph kernels for predicting the next graph in a dynamically changing series of graphs. Using a similar setting, Anil et al.~\cite{Anil2014} propose spectral graph kernels to predict the evolution of social networks. General spatio-temporal convolution kernels for trajectory data of simultaneously moving objects were introduced in~\cite{Knauf2016}. 
Wang et al.~\cite{wang2018time} study the classification of temporal graphs in which both vertex and edge sets can change over time. 
Here, information propagation is considered as a sequence of individual graphs, where the number of vertices is increasing when an information outburst occurs.
They aim to identify such outbursts of information propagation. 
However, to the best of our knowledge, no graph kernels have been suggested that take temporality of edges and labels, as well as dissemination processes on graphs into account.

An extensive overview of temporal graphs, their static representations, and temporal walks can be found in \cite{holme2015modern,michail2016introduction}.
In~\cite{Nguyen2018} temporal random walks are used to obtain node embeddings for link prediction in evolving networks. 
In \cite{grindrod2011communicability}, the Katz centrality is extended to temporal graphs using temporal walks.
Holme~\cite{holme2013epidemiologically} examines how temporal graphs can be used for epidemiological models.
The author evaluates different static representations of temporal graphs by %
comparing predicted and simulated epidemic outbreak sizes.

Identifying vertices that play an important role in dissemination processes has also been studied. For example, Leskovec et al.~\cite{leskovec2007cost} study the problem of placing sensors in a water distribution network to quickly detect contaminants. Recently, methods for dynamic graphs were proposed where edges may be added to the graph as time progresses. All of these approaches focus on link prediction in single graphs, see, e.g.,~\cite{Nguyen2018, Tri+2019}.  
Graph neural networks~\cite{Gil+2017} emerged as an alternative for graph classification. Our new approaches can be combined with neural architectures, e.g., see~\cite{Mor+2019}.

\section{Preliminaries}\label{sec:preliminaries}
A \emph{labeled, undirected (static)} \emph{graph} $G=(V, E, l)$ consists of a finite set of vertices $V$, a finite set $E\subseteq\{\{u,v\}\subseteq V\mid u\neq v\}$ of undirected edges, 
and a labeling function 
$l \colon V \cup E \to \Sigma$ that assigns a label to each vertex or edge, where $\Sigma$ is a finite alphabet.
In a labeled, \emph{directed} graph  $E\subseteq \{(u,v)\in V\times V \mid u\neq v\}$.
We use $V(G)$ to denote the set of vertices of $G$. A (static) \emph{walk} in a graph $G$ is an alternating sequence of vertices and edges connecting consecutive vertices. 
For notational convenience we sometimes omit edges. The length of a walk $(v_1,v_2,\ldots,v_{k+1})$ is $k$.

A \emph{labeled, undirected, temporal graph} $\tg=(V, \tge, l)$ consists of a finite set of vertices $V$, a finite set $\tge$ of undirected \emph{temporal edges} $e=(\{u,v\},t)$ with $u$ and $v$ in $V$, $u\neq v$, \emph{availability time} (or \emph{time stamp}) $t \in \mathbb{N}$, and a labeling function. Here the labeling function $l \colon V\times T \to \Sigma$ assigns a label to each vertex at each time step $t\in T = \{1,\ldots,t_{\max}+1\}$ with $t_{\max}$ being the largest time stamp of any $e\in \tge$. 
Note that for a temporal graph the number of edges is not polynomially bounded by the number of vertices. For $v\in V$ let $T(v)$ be the set of availability times of edges incident to $v$.
For convenience, we regard the set of availability times as a sequence that is ordered by the canonical ordering on the natural numbers.
The bijection $\tau_v \colon T(v) \rightarrow \{1,\ldots,|T(v)|\}$ assigns to each time its position in the ordered sequence of \mbox{$T(v)$ for $v\in V$}. 
The degree $d(v)$ of a vertex $v$ in a temporal graph is the sum of the numbers of edges incident to $v$ over all time steps.

\subsection{Kernels for Static Graphs}\label{sec:static_gk}

A \emph{kernel} on a non-empty set $\mathcal{X}$ is a positive semidefinite function 
$\kappa \colon \mathcal{X} \times \mathcal{X} \to \mathbb{R}$.
Equivalently, a function $\kappa$ is a kernel if there is a \emph{feature map} 
$\phi \colon \mathcal{X} \to \mathcal{H}$ to a Hilbert space $\mathcal{H}$ with inner product 
$\langle \cdot, \cdot \rangle$, such that 
$\kappa(x,y) = \langle \phi(x),\phi(y) \rangle$ for all $x$ and $y$ in $\mathcal{X}$.
Let $\mathcal{G}$ be the set of all graphs, then a function $\mathcal{G} \times \mathcal{G} \to \mathbb{R}$ is a \emph{graph kernel}.
We briefly summarize two well-known kernels for static graphs.

\paragraph{Random walk kernels} measure the similarity of two graphs by counting their (weighted) common walks~\cite{Gaertner2003,Sugiyama2015}. 
For a walk $w=(v_1,e_1,\dots,v_{k+1})$ let $L(w)=(l(v_1),l(e_1),\dots,l(v_{k+1}))$ denote the labels encountered on the walk. Two walks $w_1$ and $w_2$ are considered to be common if $L(w_1) = L(w_2)$.
We consider the \emph{$k$-step random walk kernel} $\kappa^k_{\text{RW}}(G,H) = \langle \phi_{\text{RW}}(G), \phi_{\text{RW}}(H)\rangle$ counting walks up to length $k$. 
Here, the feature map $\phi_{\text{RW}}$ is defined by $\phi_{\text{RW}}(G)_s = |\{w \in \mathcal{W}_\ell(G) \mid s = L(w)\}|$, where $s \in \Sigma^{2\ell+1}$ and $\mathcal{W}_\ell(G)$ is the set of walks in $G$ of length $\ell \in \{0,\dots,k\}$.

\paragraph{Weisfeiler--Lehman subtree kernels} are based on the well-known \new{color refinement algorithm} for isomorphism testing~\cite{She+2011}: Let $G$ and $H$ be graphs, and $l$ be a labeling function $V(G) \cup V(H)\to \Sigma$. In each iteration $i \geq 0$, the algorithm computes a labeling function $l^i \colon V(G) \cup V(H)\to \Sigma$, where $l^0 = l$. Now in iteration $i>0$, we set 
$l^{i}(v) = (l^{i-1}(v),\{\!\!\{ l^{i-1}(u) \mid u \in N(v) \}\!\!\})$
for $v$ in $V(G) \cup V(H)$. In practice, one maps the above pair to a unique element from $\Sigma$. The idea of the \emph{Weisfeiler--Lehman subtree kernel}~\cite{She+2011} is to compute the above algorithm for $h \geq 0$ iterations and after each iteration $i$ compute a feature map $\phi^i(G)$ in $\mathbb{R}^{|\Sigma_i|}$ for each graph $G$, where $\Sigma_i \subseteq \Sigma$ denotes the image of $l^i$. Each component $\phi^i(G)_{c}$ counts the number of occurrences of vertices labeled with $c$ in $\Sigma_i$. The overall feature map $\phi_{\text{WL}}(G)$ is defined as the concatenation of the feature maps of all $h$ iterations. 
Then, the Weisfeiler--Lehman subtree kernel for $h$ iterations is $\kappa^h_{\mathrm{WL}}(G,H) = \langle \phi_{\text{WL}}(G), \phi_{\text{WL}}(H) \rangle$. 
This kernel can also be interpreted in terms of walks. A label $l^i(v)$ represents the unique rooted tree of height $i$ obtained by simultaneously taking all possible walks  of length $i$ starting at $v$, where repeated vertices visited in the past are treated as distinct. 

\section{A Framework for Temporal Graph Kernels}\label{sec:temporalk}
\begin{table*}[t]\centering	\renewcommand{\arraystretch}{1.1}
\caption{Overview of the trade-offs of the proposed transformations. 
	The third column describes the ability of the approaches to take waiting times into account: \xmark~not supported, \ding{72}~always,  \cmark~approach is flexible.
} 
\label{table:comp}
\resizebox{.95\textwidth}{!}{ 	\renewcommand{\arraystretch}{1.2}
	\begin{tabular}{lccc}
		\textbf{Transformation} 		& \textbf{Preserves information} 	& \textbf{Waiting times} & \textbf{Size of static graph}   \\\toprule 
		Reduced Graph Representation 	& \xmark							&  \xmark  			 & $\mathcal{O}(|V|^2)$  		  \\ 
		Directed Line Graph Expansion 	& \cmark 							&  \cmark                & $\mathcal{O}(|\tge|^2)$ 		  \\ 
		Static Expansion 				& \cmark 							&  \ding{72}  				 & $\mathcal{O}(|\tge|)$ 	 
	\end{tabular} }
\vspace{-10pt}
\end{table*}

In the following temporal graphs are mapped to a static graphs such that conventional static kernels can be applied, e.g., the random walk or Weisfeiler-Lehman subtree kernel. 
To catch temporal information, we use \emph{temporal walks}, 
which are time respecting walks. 
That is, the traversed edges along a temporal walk have strictly increasing availability times.  
We assume that traversing an edge in a temporal walk does need one unit of time and is not possible instantaneously.  
\begin{definition}\label{definition:tw}
	A \emph{temporal walk} of length $\ell$ is an alternating sequence of vertices and temporal edges 
	$\left(v_1,e_1=(\{v_1,v_2\},t_1),v_2,\dots,e_\ell=(\{v_\ell,v_{\ell+1}\},t_\ell),v_{\ell+1}\right)$ 
	such that $t_{i} < t_{i+1}$ for $1\leq i <\ell$. 
	For a temporal walk the \emph{waiting time} at vertex $v_i$ with $1 < i \leq \ell$ is $t_i-(t_{i-1}+1)$.
    The set of temporal walks (of length $\ell$) in a temporal graph $\tg$ is denoted by $\mathcal{W}^\mathit{\text{tmp}}(\tg)$ ($\mathcal{W}^\mathit{\text{tmp}}_\ell(\tg)$).
    Finally, we define the function $L$ that maps a temporal walk $w$
    to the label sequence $L(w)=(l(v_1,t_1),l(v_2,t_1+1),l(v_2,t_2),l(v_3,t_2+1),\dots,l(v_\ell,t_\ell),l(v_{\ell+1},t_\ell+1))\text{.}$
\end{definition}
Temporal walks enable us to gain insights into the interpretation of the derived temporal graph kernels whenever the static graph kernel can be understood in terms of walks. This is natural in the case of random walk kernels, but also for the widely-used Weisfeiler-Lehman subtree kernel, cf.\@ Section~\ref{sec:static_gk}. We  introduce three approaches, that differ in the size of the resulting static graph and in their ability to preserve temporal information as well as to model waiting times, see \Cref{table:comp} for an overview.

\subsection{Reduced Graph Representation}\label{subsec:rdrep}
First, we propose a straight-forward approach to incorporate temporal information in static graphs.
In a temporal graph $\tg=(V, \tge, l)$ a pair of vertices may be connected with multiple edges each with a different availability time.
In this case, we only preserve the edge with the earliest availability time and delete all other edges. %
We obtain the subgraph $\tg'=(V, \tge', l)$ with $\tge'\subseteq\tge$.
From this we construct a static, labeled, undirected graph $\mathit{RG}(\tg)=(V,E,l_s)$ by inserting an edge $e=\{u,v\}$ into $E$  for every temporal edge $e'=(\{u,v\},t')\in \tge'$ . 
We set the new static edge labels  $l_s(e)$ to the number of the position $\tau(t')$ in the ordered sequence of all (remaining)  availability times $t'$ in $\tge'$.
Next, the temporal development of the dissemination is encoded using the vertex labels.  
Therefore, if the label of a vertex $v\in V$ in $\tg$ stays constant over time, we set $l_s(v)=0$. 
For the remaining vertex labels we take the ordered sequence $T_V$ of all points in time when any vertex label changes for the first time.
Then, for each vertex changing its label for the first time at time $t_l$, we set $l_s(v)=\tau(t_l)$, where $\tau(t_l)$ denotes the position of $t_l$ in the  sequence $T_V$.

Applying this procedure results in the reduced graph representation $\mathit{RG}(\tg)=(V, E, l_s)$. 
Clearly, the transformation may lead to a loss of information.
However, notice that $\mathit{RG}(\tg)$ is a simple, undirected graph with at most one edge between each pair of vertices. 
Hence, its number of edges is bounded by $|V|^2$, which can be much smaller than $|\tge|$.

\subsection{Directed Line Graph Expansion}\label{subsec:dlk}
In order to avoid a loss of information, %
we propose to represent temporal graphs by 
directed static graphs that are capable of fully encoding temporal information.

\begin{definition}[Directed line graph expansion] 
    Given
	 a temporal graph $\tg=(V,\tge,l)$, the \emph{directed line graph expansion}
	$\mathit{DL}(\tg)=(V',E',l')$ is the directed graph, where  
	every temporal edge $(\{u,v\}, t)$ is represented by two vertices $n^t_{\overrightarrow{uv}}$
	and $n^t_{\overrightarrow{vu}}$, and there is an edge from  
    $n^t_{\overrightarrow{uv}}$ to  $n^{s}_{\overrightarrow{xy}}$
    if $v=x$ and $t<s$.
    For each vertex $n^{t}_{\overrightarrow{uv}}$, we set the label $l'(n^{t}_{\overrightarrow{uv}}) = (l(u,t), l(v,t+1))$.
\end{definition}
\Cref{fig:example_ld} shows an example of the trans\-for\-ma\-tion for the graph shown in \cref{fig:example_tempG}. 
The following lemma relates temporal walks in a temporal graph and the static walks in its directed line graph expansion.
\begin{proposition}\label{lem:tmp_walk_to_walk}
    Let $\tg$ be a temporal graph and $\ell \geq 0$.
    The walks in $\mathcal{W}_\ell(\mathit{DL}(\mathbf{G}))$ are in one-to-one 
    correspondence with the temporal walks in $\mathcal{W}^\text{tmp}_{\ell+1}(\tg)$.
\end{proposition}
\begin{proof}
    Let $(v_1, (\{v_1,v_2\}, t_1), v_2, \ldots, (\{v_\ell,v_{\ell+1}\}, t_\ell), v_{\ell+1})$ be a temporal walk of length
    $\ell \geq 2$ in $\tg$. 
    Then $\left(n^{t_1}_{\overrightarrow{v_1v_2}}, n^{t_2}_{\overrightarrow{v_2v_3}}, \ldots, n^{t_\ell}_{\overrightarrow{v_\ell v_{\ell+1}}}\right)$ is a walk of length $\ell-1$ in $\mathit{DL}(\tg)$, since the vertices represent temporal edges of $\tg$ that are consecutively connected by directed edges in $\mathit{DL}(\tg)$ provided that $t_i < t_{i+1}$ for  $i \in \{1,\ldots, \ell-1\}$. This holds for every temporal walk.
    
    Vice versa, let $\left(n^{t_1}_{\overrightarrow{v_1v_2}}, n^{t_2}_{\overrightarrow{v_2v_3}}, \ldots, n^{t_\ell}_{\overrightarrow{v_\ell v_{\ell+1}}}\right)$ be a walk in $\mathit{DL}(\tg)$. 
    Due to the construction of $\mathit{DL}(\tg)$ the time stamps satisfy $t_1 < t_2 < \cdots < t_\ell$ 
    and we can construct a unique temporal walk in $\tg$ from the sequence of vertices 
    as above.
\end{proof}
Note that for $\ell=0$ the two vertices $n^{t}_{\overrightarrow{uv}}$ 
and $n^{t}_{\overrightarrow{vu}}$ representing the same temporal edge $(\{u,v\},t)$
correspond to two different temporal walks traversing the edge in different directions.
In \cref{fig:example_ld} the walk $(n^2_{\overrightarrow{ca}},n^3_{\overrightarrow{ab}},n^7_{\overrightarrow{bc}})$ in $DL(\tg)$ of length $2$ corresponds to the temporal walk $(c,(\{c,a\},2),a,(\{a,b\},3),b,(\{b,c\},7),c)$ of length $3$ in the temporal graph $\tg$. 

The vertex labeling of $\mathit{DL}(\tg)$ is sufficient to encode 
all the label information of the temporal graph $\tg$, i.e., two temporal walks 
exhibit the same label sequence (according to the function $L$ in \Cref{definition:tw}) if and only if the corresponding walks in the 
directed line graph expansion have the same label sequence. 
Therefore, all static kernels that can be interpreted in terms of walks are
lifted to temporal graphs and the concept of temporal walks by applying them to 
the directed line graph expansion.
Moreover, the directed line graph expansion supports to take waiting times into 
account by annotating an edge 
$\left(n^s_{\overrightarrow{uv}}, n^{t}_{\overrightarrow{vw}}\right)$ with the 
waiting time $t-s-1$ at $v$. 
We proceed by studying basic structural properties of the directed line graph
expansion.
\begin{proposition}\label{proposition:dlsize}
	Let $\tg=(V,\tge)$ be a temporal graph and $\mathit{DL}(\tg)=(V',E')$ its
	directed line graph expansion. Then, $|V'| = 2|\tge|$ and 
	\vspace{-4pt}$$|E'| \leq -|\tge|+\frac{1}{2} \sum_{v \in V} d^2(v) = \mathcal{O}(|\tge|^2)\text{,}\vspace{-4pt}$$
	where $d(v)$ denotes the degree of $v$ in $\tg$. 
\end{proposition}
\begin{proof}
    The number of vertices directly follows from the construction. The directed line
    graph expansion $\mathit{DL}(\tg)$ of the temporal graph $\tg$ is closely related
    to the classical directed line graph of a modified copy of $\tg$ obtained by deleting 
    all time stamps (leading to indistinguishable parallel edges) and replacing all 
    undirected edges by two directed copies. 
    We use the following classical result~\cite{Harary1960}.
    For a directed graph $G=(V,E)$ the number of edges in the directed line graph is
    $
    \sum_{v \in V} d^-(v)\cdot d^+(v)$, where $d^-(v)$ is the in- and $d^+(v)$ the outdegree of $v$.
    After replacing all undirected edges by two directed copies the in- and outdegrees are equal for each vertex.
    Furthermore, the directed line graph expansion of $\tg$ has at most half the edges, because for the vertices 
    $n^t_{\overrightarrow{uv}}$, $n^{s}_{\overrightarrow{vw}}$ and $n^t_{\overrightarrow{vu}}$,
    $n^{s}_{\overrightarrow{wv}}$ at most one time constraint $s<t$ or $s>t$ can be
    satisfied. 
    Finally, between $n^{t}_{\overrightarrow{vw}}$ and $n^t_{\overrightarrow{wv}}$ there is no edge, therefore we subtract $|\tge|$ and the result follows.
\end{proof}
Since a cycle in the directed line graph expansion would correspond to a cyclic 
sequence of edges with strictly increasing time stamps, the directed line graph 
$\mathit{DL}(\tg)$ of a temporal graph $\tg$ is acyclic.
Consequently, the maximum length of a walk in the directed line graph expansion
is bounded. Therefore, there is no need to 
down-weight walks with increasing length to ensure convergence, which avoids the problem of
\emph{halting}~\cite{Sugiyama2015}.

\subsection{Static Expansion Kernel}\label{subsec:sek}
A disadvantage of the directed line graph expansion is that it may lead to a quadratic blowup w.r.t.\ the number of temporal edges. 
Here, we propose an approach that utilizes the \emph{static expansion} of a temporal graph resulting in a static graph of linear size.
The static expansion $\mathit{SE}(\tg)$ of a temporal graph $\tg$ is a static, directed and acyclic graph that 
contains the temporal information of $\tg$.
Similar approaches for static expansions have been used to solve a variety of problems on temporal graphs, see, e.g.,~\cite{michail2016traveling}. 

For $\tg=(V,\tge, l)$ we construct $\mathit{SE}(\tg)=(U,E,l')$ with
$U$ being a set of \emph{time-vertices}. Each time-vertex $(v,t)\in U$ represents a vertex $v\in V$ at time $t$. 
Time-vertices are connected by directed edges that mirror the flow of time, i.e., if an edge from $(v,t)\in U$ to $(u,s)\in U$ exists, then $t<s$.
Because edges in $\tge$ are non-directed, the transformation has to consider both possible directions. It follows, that for each temporal edge $e \in \tge$, we introduce at most four time-vertices that represent the start and end points of $e$.
Next, we add edges that correspond to temporal edges in $\tge$, and additional edges  that represent possible waiting times at a vertex.
\begin{definition}[Static expansion]\label{def:staticexpansion}
	For a temporal graph $\tg=(V,\tge, l)$, we define the static expansion as a labeled, directed graph $\mathit{SE}(\tg)=(U,E, l')$, with vertex set
	$U=\{(u,t), (v,t), (u,t+1), (v,t+1) \mid  (\{u,v\},t) \in \tge \}\text{, and}$ edge set
	$E=E_N \cup E_{W_1} \cup E_{W_2}$, where %
	\begin{align*}
	E_N\,\,&=\big\{\big((u,t), (v,t+1)\big), \big((v,t),(u,t+1)\big) \mid \\&\hspace{9mm} (\{u,v\},t) \in \tge \big\} 	\text{,} \\
	E_{W_1}&=\big\{\big((w,i+1), (w,j)\big)\mid (w,i+1), (w,j)\in U\text{, }
	\\&i,j\in T(w),\tau_w(i)+1 = \tau_w(j)\text{ and } i+1<j \big\}
	\text{, and} \\
	E_{W_2}&=\big\{\big((w,i), (w,j)\big)\mid (w,i), (w,j)\in U \text{, }
	\\&i,j\in T(w),\tau_w(i)+1 = \tau_w(j) \text{ and } i<j \big\}
	\text{.} %
	\end{align*}
	For each time-vertex $(w,t)\in U$, we set $l'((w,t))= l(w, t)$. 
	For each edge in $e\in E_N$, we set $l'(e)= \eta$,
	and for each edge in $e\in E_{W_1} \cup E_{W_2}$, we set $l'(e)= \omega$.
\end{definition}
\Cref{fig:example_tg} shows an example of the transformation.
\begin{figure}[t]
	\centering
	\begin{subfigure}{0.24\textwidth}
		\centering
		\begin{tikzpicture}[scale=0.6]
		\node[vertexa] (1) at (0,0) [label=left:\scriptsize$a$]{};
		\node[vertexa] (3) at (3,-3) [label=right:\scriptsize$c$]{};
		\node[vertexb] (4) at (0,-3) [label=left:\scriptsize$b$]{};
		\path[->] 
		(1)  edge[edge]  node[midway,left] {\scriptsize$3$}   (4)
		(1)  edge[edge]  node[midway,above] {\scriptsize$2$}   (3)
		(3)  edge[edge]  node[midway,below,sloped] {\scriptsize$7$}   (4);	
		
		\end{tikzpicture}
		\caption{}
                     \label{fig:example_tempG}
    \vspace{-10pt}
	\end{subfigure}
	\begin{subfigure}{0.24\textwidth}
		\centering
		\begin{tikzpicture}[scale=0.88]

		\node[vertexldg] (1) at (0,0) [label=left:\scriptsize$\textcolor{gray}{a}$]{};
		\node[vertexldg] (3) at (3,-3) [label=right:\scriptsize$\textcolor{gray}{c}$]{};
		\node[vertexldg] (4) at (0,-3) [label=left:\scriptsize$\textcolor{gray}{b}$]{};
		
		\node[vertexld] (a) at (1.8,-1.3) {}; 
		\node[vertexld] (b) at (1.3,-1.8) {}; 
		\node[vertexld] (d) at (-0.35,-1.5) {}; 
		\node[vertexld] (c) at (0.25,-1.5) {}; 
		\node[vertexld] (e) at (1.5,-3.25) {}; 
		\node[vertexld] (f) at (1.5,-2.75) {}; 
		
		\path[->] 
		(1)  edge[bend left=15,ldgedge]  node[pos=.2,left] {\scriptsize$3$}   (4)
		(1)  edge[bend left=15,ldgedge]  node[pos=.1,above] {\scriptsize$2$}   (3)
		(3)  edge[bend left=15,ldgedge]  node[pos=.2,below,sloped] {\scriptsize$7$}   (4)	
		(4)  edge[bend left=15,ldgedge]  node[pos=.2,left] {\scriptsize$3$}   (1)
		(3)  edge[bend left=15,ldgedge]  node[pos=.1,above] {\scriptsize$2$}   (1)
		(4)  edge[bend left=15,ldgedge]  node[pos=.2,below,sloped] {\scriptsize$7$}   (3)
		
		(a)  edge[bend left=60,dedge,line width=0.3mm]  node[pos=.2,left] {}   (e)
		(b)  edge[bend right=15,dedge,line width=0.3mm]  node[pos=.2,above] {}   (c)
		(c)  edge[bend right=15,dedge,line width=0.3mm]  node[pos=.2,below,sloped] {}   (f);
		
		\node[vertexld] (a) at (1.8,-1.3) {\scriptsize$n^2_{\overrightarrow{ac}}$}; 
		\node[vertexld] (b) at (1.3,-1.8) {\scriptsize$n^2_{\overrightarrow{ca}}$};
		\node[vertexld] (d) at (-0.35,-1.5) {\scriptsize$n^3_{\overrightarrow{ba}}$}; 
		\node[vertexld] (c) at (0.25,-1.5) {\scriptsize$n^3_{\overrightarrow{ab}}$};
		\node[vertexld] (e) at (1.5,-3.25) {\scriptsize$n^7_{\overrightarrow{cb}}$}; 
		\node[vertexld] (f) at (1.5,-2.75) {\scriptsize$n^7_{\overrightarrow{bc}}$}; 
		
		\end{tikzpicture}
		\caption{}
		\label{fig:example_ld}
        \vspace{-10pt}
	\end{subfigure}
	\begin{subfigure}{0.4\textwidth}
		\centering
		\begin{tikzpicture}[scale=0.75]

		\node[vertexa] (a2) at (0.5,0)[label=above:\scriptsize${(a,2)}$]{};
		\node[vertexa] (a3) at (2,0) [label=above:\scriptsize${(a,3)}$]{};
		\node[vertexb] (b3) at (2,-1.5) [label=below:\scriptsize${(b,3)}$]{};
		\node[vertexa] (c2) at (0.5,-3) [label=below:\scriptsize${(c,2)}$]{};
		\node[vertexa] (c3) at (2,-3) [label=below:\scriptsize${(c,3)}$]{};
		\node[vertexa] (a4) at (3.5,0) [label=above:\scriptsize${(a,4)}$]{};
		\node[vertexb] (b4) at (3.5,-1.5) [label=above:\scriptsize${(b,4)}$]{};
		\node[vertexa] (c7) at (5,-3) [label=below:\scriptsize${(c,7)}$]{};
		\node[vertexa] (c8) at (6.5,-3) [label=below:\scriptsize${(c,8)}$]{};	
		\node[vertexb] (b7) at (5,-1.5) [label=above:\scriptsize${(b,7)}$]{};
		\node[vertexb] (b8) at (6.5,-1.5) [label=above:\scriptsize${(b,8)}$]{};
		\path[->] 
		(a2)  edge[wedge]  node[midway,above,sloped] {\scriptsize$\omega$}   (a3)
		(b3)  edge[wedge,bend right=35]  node[midway,below,sloped] {\scriptsize$\omega$}   (b7)
		(b4)  edge[wedge]  node[midway,above,sloped] {\scriptsize$\omega$}   (b7)
		(c2)  edge[wedge,bend right=45]  node[midway,above,sloped] {\scriptsize$\omega$}   (c7)
		(c3)  edge[wedge]  node[midway,above,sloped] {\scriptsize$\omega$}   (c7)
		(a2)  edge[dedge]  node[midway,above,pos=0.3] {\scriptsize$\eta$}   (c3)
		(c2)  edge[dedge]  node[midway,below,pos=0.3] {\scriptsize$\eta$}   (a3)
		(a3)  edge[dedge]  node[midway,above,pos=0.3] {\scriptsize$\eta$}   (b4)
		(b3)  edge[dedge]  node[midway,below,pos=0.3] {\scriptsize$\eta$}   (a4)
		(b7)  edge[dedge]  node[midway,above,pos=0.3] {\scriptsize$\eta$}   (c8)
		(c7)  edge[dedge]  node[midway,below,pos=0.3] {\scriptsize$\eta$}   (b8);		
		\end{tikzpicture}
		\caption{}
		\label{fig:example_tg}
	\end{subfigure}
	\caption{(a) A temporal graph $\tg$ with (static) red and black vertex labels. (b) Directed line graph expansion $\mathit{DL}(\tg)$ (edges are in solid and vertex labels are omitted) (c) Static expansion $\mathit{SE}(\tg)$.}
\vspace{-15pt}
\end{figure}
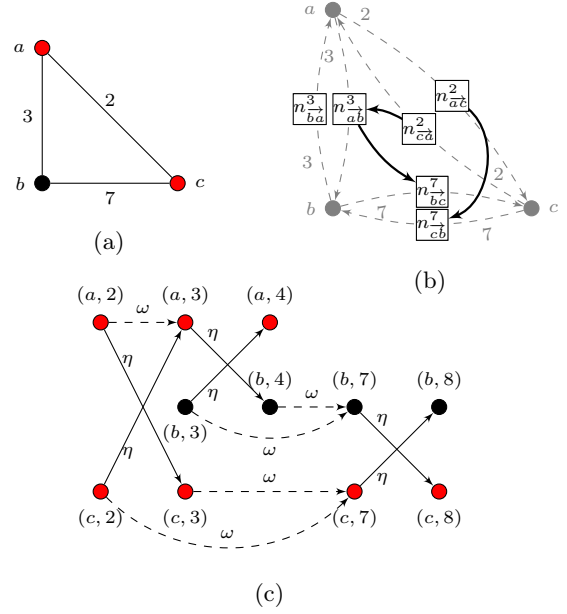
    Notice that the label sequence of a static walk in $SE(\tg)$ encodes the times of waiting.
	Since for all edges $((u,t_1),(v,t_2))\in E$ it holds that $t_1<t_2$, the resulting graph is acyclic.
	Finally, we have the following result.
	\begin{proposition}\label{theorem:se_size}
		The size of the static expansion $\mathit{SE}(\tg)=(U,E,l')$ of a temporal graph $\tg$ is in $\mathcal{O}(|\tge|)$. 
	\end{proposition}
\begin{proof}
    The size of $U$ is bounded by $4\cdot|\tge|$, because at most four vertices for each temporal edge $e\in \tge$ are inserted.
    For each $e\in \tge$, there are at most six edges in $E$. Two edges that represent using $e$ at time $t$ and four possible waiting edges, one at each vertex inserted for edge $e$.
    Consequently, $|E|\in\mathcal{O}(|\tge|)$.\hspace{\columnwidth}
\end{proof}
\section[Approximation for the directed line graph representation]{Approximation for the Directed Line Graph Representation\footnote{Section 4 has been updated and improved compared to the conference version \cite{oettershagen2020temporal} of this paper.}}\label{sec:approx}
Although the directed line graph expansion preserves the temporal information and is able to model waiting times, see~\Cref{table:comp} and~\Cref{lem:tmp_walk_to_walk}, the construction may lead to a blowup in graph size. 
Hence, we propose a stochastic variant based on sampling temporal walks with provable approximation guarantees.

Let $\tg=(V,\tge,l)$ be a temporal graph, the algorithm approximates the normalized feature vector $\widehat{\phi}^k_{\text{RW}}(\tg) = \phi^k_{\text{RW}}(\tg)/ \norm{\phi^k_{\text{RW}}(\tg)}_1$, resulting in the normalized kernel $\widehat{\kappa}^k_{\text{RW}}(\tg_1,\tg_2)=\langle\widehat{\phi}^k_{\text{RW}}(\tg_1), \widehat{\phi}^k_{\text{RW}}(\tg_2) \rangle$ for two temporal graphs $\tg_1$ and $\tg_2$. 
The algorithm starts by sampling $S$ temporal walks at random (with replacement) from all possible walks in $\tg$, where the exact value of $S$ will be determined later. We compute a histogram $\widetilde{\phi}^k_{\text{RW}}(\tg)$, where each entry $\widetilde{\phi}^k_{\text{RW}}(\tg)_s$ counts the number of temporal walks $w$ with $L(w)=s$ encountered during the above procedure, normalized by $\nicefrac{1}{|S|}$. 
\Cref{alg:as} shows the pseudocode. 
\begin{algorithm}\mbox{\hfill}
    \\\textbf{Input:} A temporal graph $\tg$, a walk length $k\in\mathbb{N}$, $S\in\mathbb{N}$.
    \\\textbf{Output:} A feature vector $\widetilde{\phi}^k_{\text{RW}}(\tg)$ of normalized temporal walk counts.
    \begin{algorithmic}[1]
        \State Initialize feature vector $\widetilde{\phi}^k_{\text{RW}}(\tg)$ to a vector of zeros
        \ParForAllLoop{$1,\ldots,S$}
        \State sample $k$-step walk $w$
        \State $\widetilde{\phi}^k_{\text{RW}}(\tg)_{L(w)} \gets \widetilde{\phi}^k_{\text{RW}}(\tg)_{L(w)} + \nicefrac{1}{S}$
        \EndParForAllLoop
        \State \Return$\widetilde{\phi}^k_{\text{RW}}(\tg)$ 
    \end{algorithmic}
    \caption{}
    \label{alg:as}
\end{algorithm}

We get the following result, showing that~\Cref{alg:as} can approximate the normalized (temporal) random-walk kernel $\widehat{\kappa}^k_{\text{RW}}(\tg_1, \tg_2)$ up to an additive error.  
\begin{theorem}\label{tgclass:theorem:approx}
    Let $\mathcal{G}$ be a set of temporal graphs with label alphabet $\Sigma$. Moreover, let $k > 0$, and let $\Gamma(\Sigma,k)$ denote an upperbound for the number of temporal walks of length $k$ with labels from $\Sigma$. By setting 
    \begin{equation}\label{samp}
    S = \frac{\log( 2 \cdot  |\mathcal{G}| \cdot \Gamma(\Sigma,k) \cdot 1/\delta)}{ 2 (\nicefrac{\lambda}{\Gamma(\Sigma,h)})^2  },
    \end{equation} 
    \Cref{alg:as} approximates the normalized temporal random walk kernel $\widehat{\kappa}_{\text{RW}}$ with probability $(1-\delta)$, such that 
    \begin{equation*}
    \sup_{\tg_1, \tg_2 \in \mathcal{G}} \left|    \widehat{\kappa}^k_{\text{RW}}(\tg_1, \tg_2) - \langle \widetilde{\phi}^k_{\text{RW}}(\tg_1), \widetilde{\phi}^k_{\text{RW}}(\tg_2)  \rangle  \right| \leq 3\lambda.
    \end{equation*}
\end{theorem}
\begin{proof}
    First, by an application of the Hoeffding bound together with the Union bound, it follows that by setting 
    \begin{equation*}
    S = \frac{\log( 2 \cdot  |\mathcal{G}| \cdot \Gamma(\Sigma,k) \cdot 1/\delta)}{ 2\varepsilon^2  },
    \end{equation*}
    it holds that with probability $1-\delta$,
    \begin{equation*}
    \left| \widehat{\phi}^k_{\text{RW}}(\tg_i)_j - \widetilde{\phi}^k_{\text{RW}}(\tg_i)_j \right| \leq\varepsilon,
    \end{equation*} 
    for all $j$ for $1\leq  j \leq \Gamma(\Sigma,k)$, and all temporal graphs $\tg_i$ in $\mathcal{G}$. Let $\tg_1$ and $\tg_2$ in $\mathcal{G}$, then
    \begin{equation*}
    \begin{split}
    \widetilde{\kappa}^k_{\text{RW}}(\tg_1,\tg_2) &= \left\langle \widetilde{\phi}^k_{\text{RW}}(\tg_1), \widetilde{\phi}^k_{\text{RW}}(\tg_2) \right\rangle\\
    &\leq \smashoperator{\sum^{\Gamma(\Sigma,h)}_{i = 1}} \mleft(\widehat{\phi}^k_{\text{RW}}(\tg_1)_i \cdot\widehat{\phi}^k_{\text{RW}}(\tg_2)_i\mright) \\
    &\hspace{-8mm}+ \varepsilon \cdot \smashoperator{\sum^{\Gamma(\Sigma,h)}_{i = 1}} \mleft(\widehat{\phi}^k_{\text{RW}}(\tg_1)_i + \widehat{\phi}^k_{\text{RW}}(\tg_2)_i \mright) + \smashoperator{\sum^{\Gamma(\Sigma,h)}_{i = 1}} \varepsilon^2\\
    &\hspace{-8mm}\leq \widehat{\kappa}^k_{\text{RW}}(\tg_1,\tg_2) + 2 \Gamma(\Sigma,h)  \cdot \varepsilon + \Gamma(\Sigma,h)\cdot \varepsilon\,.
    \end{split}
    \end{equation*}
    The last inequality follows from the fact that the components of $\widehat{\phi}^k_{\text{RW}}(\cdot)$ are in $[0,1]$. The result then follows by setting $\varepsilon = \nicefrac{\lambda}{\Gamma(\Sigma,h)}$.
\end{proof}
\Cref{alg:as} can be easily modified for the static (non-temporal) graphs and applied to all three of our kernel approaches.

So far, we assumed that we have access to an oracle to sample the temporal walks uniformly at random from the temporal graph $\tg$.
One possible practical way to obtain the uniformly sampled temporal walks is to use \emph{rejection sampling}. 
To this end, we start by uniformly sampling a random edge $e_1=(\{u,v\},t_1)\in \tge$ and extend it to a $k$-step temporal random walk $w$
starting with probability $\nicefrac{1}{2}$ at node $u$ or node $v$, respectively, and at time  $t_1+1$.
At each step, we choose an incident edge uniformly from all edges with availability time not earlier than the arrival time at the current vertex. 
The probability of choosing a temporal walk $w=(v_1,e_1,\ldots,v_i,e_i=(\{v_i,v_{i+1}\},t_i),\ldots,e_k,v_{k+1})$ from all possible temporal walks of length $k$ is then $P_w=\frac{1}{2\cdot|\tge|}\cdot \Pi_{i=2}^{k}\frac{1}{d(v_i,t_{i-1}+1)}$ with $d(v_i,t_{i-1}+1)$ being the number of incident edges to $v_i$ 
with availability time at least $t_{i-1}+1$, i.e., the arrival time at $v_i$. 
On the other hand, a lower bound for the probability of any $k$-step temporal walk is $P_{min}=\frac{1}{|V|}\cdot \frac{1}{\Delta^{k}}$ with $\Delta$ being the maximal degree in $\tg$.
To obtain uniform probabilities for the sampled walks, we accept a walk only with probability $\frac{P_{min}}{P_{w}}$. Otherwise, we reject $w$ and restart the walk sampling procedure. 
Hence, each accepted walk has the uniform probability of $P_w \cdot \frac{P_{min}}{P_{w}}=P_{min}$.

Notice that we can forgo the rejection step in the case that the graphs for which we compute the kernel have similar biases during the walk sampling procedure.

\section{Experiments}\label{sec:experiments}
  \begin{table*}
    \begin{center}
        \caption{Statistics and properties of the data sets.  %
        }  %
        \label{table:datasets_stats2}
        \resizebox{0.65\textwidth}{!}{ 	\renewcommand{\arraystretch}{.8}
            \begin{tabular}{lc@{\hspace{3mm}}c@{\hspace{3mm}}c@{\hspace{3mm}}c@{\hspace{3mm}}c@{\hspace{3mm}}c@{\hspace{3mm}}c}\toprule
                \multirow{3}{2cm}{\vspace*{4pt}\textbf{Properties}\vspace*{4pt}}&\multicolumn{6}{c}{\textbf{Data set}}\\
                \cmidrule{2-7}
                \textbf{}               &  \textsc{Mit}    & \textsc{Highschool}  & \textsc{Infectious}&  \textsc{Tumblr}   & \textsc{Dblp}  & \textsc{Facebook}  \\ \midrule
                Size			     	&	$97$           &       $180$        &       $200$      &             $373$      &   $755$        &       $995$        \\
                $\varnothing$ $|V|$     &   $20$           &       $52.3$       &       $50$       &             $53.1$     &     $52.9$     &       $95.7$       \\ 	
                $\min|\tge|$            &   $126$          &       $286$        &       $218$      &             $96$       &     $206$      &       $176$        \\
                $\max|\tge|$            &   $3\,363$       &       $517$        &       $505$      &             $190$      &    $225$       &       $181$        \\ 
                $\varnothing$ $|\tge|$  &   $702.8$        &       $262.4$      &       $220.4$    &             $98.7$     &  $156.8$       &       $133.9$      \\ 
                $\varnothing$ $\max d(v)$  &   $680.7$     &       $92.5$       &       $43.8$     &             $24.4$     &  $26.4$        &       $21.0$       \\  %
                \bottomrule
            \end{tabular}
        }
    \end{center}	
\vspace{-10pt}
\end{table*}
\begin{table*}[t]\centering
	\caption{Classification accuracy in percent and standard deviation for the first classification task. } %
	\label{table:results}
	\resizebox{0.75\textwidth}{!}{ 	\renewcommand{\arraystretch}{0.9}
		\begin{tabular}{@{}c <{\enspace}@{}l@{\hspace{3mm}}ccccccc@{}}	\toprule
			
			& \multirow{3}{*}{\vspace*{4pt}\textbf{Kernel}}&\multicolumn{6}{c}{\textbf{Data set}}\\\cmidrule{3-8}
			&  & {\textsc{Mit}}         &  {\textsc{Highschool}}      &     {\textsc{Infectious}}       & {\textsc{Tumblr}}       & {\textsc{Dblp}} &  {\textsc{Facebook}} \\ \toprule
				
            \multirow{2}{*}{\rotatebox{90}{\small Static}}
            & \emph{Stat-RW} & $61.03$ $\sz\pm 2.4$ & $63.94$ $\sz\pm 2.3$ & $76.40$ $\sz\pm 1.5$ & $83.68$ $\sz\pm 1.4$ & $83.75$ $\sz\pm 0.8$ & $86.43$ $\sz\pm 0.4$ \\
            & \emph{Stat-WL} & $42.52$ $\sz\pm 2.6$ & $45.33$ $\sz\pm 2.7$ & $68.95$ $\sz\pm 2.0$ & $78.69$ $\sz\pm 1.5$ & $78.56$ $\sz\pm 0.8$ & $83.41$ $\sz\pm 0.6$ \\
            \cmidrule{1-8}
            \multirow{9}{*}{\rotatebox{90}{\small Temporal}}
            & \emph{RD-RW} & $66.20$ $\sz\pm 2.6$ & $90.83$ $\sz\pm 1.3$ & $90.40$ $\sz\pm 1.0$ & $76.15$ $\sz\pm 1.5$ & $91.69$ $\sz\pm 0.5$ & $83.73$ $\sz\pm 0.8$ \\
            & \emph{RD-WL} & $83.41$ $\sz\pm 0.5$ & $90.06$ $\sz\pm 0.6$ & $91.75$ $\sz\pm 1.0$ & $70.59$ $\sz\pm 0.9$ & $90.55$ $\sz\pm 0.5$ & $82.04$ $\sz\pm 0.7$ \\\addlinespace[0.8mm]
            & \emph{DL-RW} & $\textbf{92.91}$ $\sz\pm 0.9$ & $97.39$ $\sz\pm 0.7$ & $97.95$ $\sz\pm 0.4$ & $95.15$ $\sz\pm 0.6$ & $\textbf{98.86}$ $\sz\pm 0.1$ & $96.46$ $\sz\pm 0.1$ \\
            & \emph{DL-WL} & $91.68$ $\sz\pm 1.6$ & $\textbf{99.17}$ $\sz\pm 0.6$ & $\textbf{98.05}$ $\sz\pm 0.4$ & $94.19$ $\sz\pm 0.4$ & $98.49$ $\sz\pm 0.2$ & $\textbf{96.59}$ $\sz\pm 0.2$ \\\addlinespace[0.8mm]
            & \emph{SE-RW} & $88.56$ $\sz\pm 1.0$ & $96.89$ $\sz\pm 0.7$ & $97.30$ $\sz\pm 1.1$ & $\textbf{95.31}$ $\sz\pm 0.5$ & $98.46$ $\sz\pm 0.3$ & $95.68$ $\sz\pm 0.3$ \\
            & \emph{SE-WL} & $89.52$ $\sz\pm 1.7$ & $97.56$ $\sz\pm 0.7$ & $95.40$ $\sz\pm 0.7$ & $94.13$ $\sz\pm 1.0$ & $97.23$ $\sz\pm 0.2$ & $95.26$ $\sz\pm 0.3$ \\\addlinespace[0.8mm]
            & \emph{APPROX (S=50)} & $82.84$ $\sz\pm 2.0$ & $87.61$ $\sz\pm 1.7$ & $83.40$ $\sz\pm 1.6$ & $89.33$ $\sz\pm 0.7$ & $93.12$ $\sz\pm 0.5$ & $89.39$ $\sz\pm 0.5$ \\
            & \emph{APPROX (S=100)} & $85.12$ $\sz\pm 2.6$ & $90.22$ $\sz\pm 1.9$ & $91.05$ $\sz\pm 1.0$ & $90.40$ $\sz\pm 1.0$ & $95.67$ $\sz\pm 0.5$ & $92.40$ $\sz\pm 0.4$ \\
            & \emph{APPROX (S=250)} & $89.30$ $\sz\pm 2.8$ & $94.00$ $\sz\pm 1.3$ & $95.05$ $\sz\pm 0.9$ & $92.74$ $\sz\pm 0.3$ & $97.22$ $\sz\pm 0.5$ & $94.64$ $\sz\pm 0.3$ \\	
            \bottomrule  \end{tabular}}
            \vspace{10pt}
	\caption{Classification accuracy in percent and standard deviation for the second classification task.
		} %
	\label{table:results2}
	\resizebox{0.75\textwidth}{!}{ 	\renewcommand{\arraystretch}{0.9}
		\begin{tabular}{@{}c <{\enspace}@{}l@{\hspace{3mm}}ccccccc@{}}	\toprule	
			& \multirow{3}{*}{\vspace*{4pt}\textbf{Kernel}}&\multicolumn{6}{c}{\textbf{Data set}}\\\cmidrule{3-8}
			&  & {\textsc{Mit}}         &  {\textsc{Highschool}}      &     {\textsc{Infectious}}       & {\textsc{Tumblr}}       & {\textsc{Dblp}} &  {\textsc{Facebook}} \\ \toprule		
            \multirow{2}{*}{\rotatebox{90}{\small Static}}
            & \emph{Stat-RW} & $57.71$ $\sz\pm 4.0$ & $62.83$ $\sz\pm 2.9$ & $66.05$ $\sz\pm 2.4$ & $67.35$ $\sz\pm 0.9$ & $60.60$ $\sz\pm 1.4$ & $67.16$ $\sz\pm 0.6$ \\
            & \emph{Stat-WL} & $40.77$ $\sz\pm 2.5$ & $64.06$ $\sz\pm 2.4$ & $64.00$ $\sz\pm 1.3$ & $70.11$ $\sz\pm 0.5$ & $64.33$ $\sz\pm 0.8$ & $69.92$ $\sz\pm 0.3$ \\
            \cmidrule{1-8}
            \multirow{9}{*}{\rotatebox{90}{\small Temporal}}
            & \emph{RD-RW} & $60.57$ $\sz\pm 4.7$ & $80.06$ $\sz\pm 1.9$ & $74.00$ $\sz\pm 1.5$ & $69.54$ $\sz\pm 1.0$ & $64.24$ $\sz\pm 1.2$ & $66.35$ $\sz\pm 0.6$ \\
            & \emph{RD-WL} & $69.20$ $\sz\pm 1.8$ & $83.94$ $\sz\pm 0.7$ & $77.40$ $\sz\pm 1.2$ & $69.74$ $\sz\pm 0.5$ & $67.30$ $\sz\pm 0.5$ & $66.66$ $\sz\pm 0.6$ \\\addlinespace[0.8mm]
            & \emph{DL-RW} & $\textbf{82.64}$ $\sz\pm 2.1$ & $\textbf{93.44}$ $\sz\pm 1.0$ & $\textbf{88.65}$ $\sz\pm 1.2$ & $77.21$ $\sz\pm 1.0$ & $81.79$ $\sz\pm 0.9$ & $79.97$ $\sz\pm 0.5$ \\
            & \emph{DL-WL} & $36.40$ $\sz\pm 4.0$ & $89.33$ $\sz\pm 0.7$ & $78.65$ $\sz\pm 1.5$ & $78.18$ $\sz\pm 1.3$ & $76.45$ $\sz\pm 1.0$ & $79.99$ $\sz\pm 0.5$ \\\addlinespace[0.8mm]
            & \emph{SE-RW} & $57.27$ $\sz\pm 3.3$ & $93.28$ $\sz\pm 1.3$ & $87.20$ $\sz\pm 1.1$ & $78.23$ $\sz\pm 0.9$ & $\textbf{83.09}$ $\sz\pm 0.6$ & $\textbf{80.15}$ $\sz\pm 0.6$ \\
            & \emph{SE-WL} & $50.79$ $\sz\pm 4.2$ & $91.00$ $\sz\pm 1.2$ & $79.80$ $\sz\pm 1.5$ & $\textbf{80.64}$ $\sz\pm 0.7$ & $82.04$ $\sz\pm 0.6$ & $75.29$ $\sz\pm 0.3$ \\\addlinespace[0.8mm]
            & \emph{APPROX (S=50)} & $60.41$ $\sz\pm 4.0$ & $76.39$ $\sz\pm 3.0$ & $74.60$ $\sz\pm 1.8$ & $74.73$ $\sz\pm 1.3$ & $69.70$ $\sz\pm 0.8$ & $72.04$ $\sz\pm 0.7$ \\
            & \emph{APPROX (S=100)} & $65.41$ $\sz\pm 4.2$ & $83.83$ $\sz\pm 1.8$ & $75.60$ $\sz\pm 1.9$ & $76.49$ $\sz\pm 1.5$ & $74.50$ $\sz\pm 0.7$ & $73.08$ $\sz\pm 0.6$ \\
            & \emph{APPROX (S=250)} & $66.93$ $\sz\pm 2.5$ & $90.39$ $\sz\pm 1.8$ & $78.60$ $\sz\pm 2.1$ & $78.44$ $\sz\pm 1.3$ & $76.44$ $\sz\pm 0.9$ & $77.88$ $\sz\pm 0.6$ \\
            \bottomrule  
        \end{tabular}}
\end{table*}
To evaluate our proposed approaches and investigate their benefits compared to static graph kernels, we address the following questions:
\begin{description}
	\item[Q1] How well do our temporal kernels compare to each other and static approaches in terms of \textbf{(a)} accuracy and \textbf{(b)} running time?
	\item[Q2] How does the approximation for the directed line graph approach compare to the exact algorithm?
	\item[Q3] How is the classification accuracy affected by incomplete knowledge of the dissemination process?
\end{description}
\subsection{Data Sets}

We used the following real-world temporal graph data sets representing different types of social interactions. 
\begin{dsitemize}\itemsep-0pt
	\item \emph{Infectious} and \emph{Highschool:}
	Two data sets from the \emph{SocioPatterns} project.\footnote{\url{http://www.sociopatterns.org}}
	The Infectious graph represents face-to-face 
	contacts between visitors of the exhibition \textit{Infectious: Stay Away}~\cite{Isella2011}.
	The Highschool graph is a contact network and represents interactions between students in twenty second intervals over seven days. 
	\item \emph{MIT:} A temporal graph of interactions among students of the Massachusetts Institute of Technology~\cite{konect:eagle06}. 
	\item \emph{Facebook} and \emph{Tumblr:} 
	The first graph is a subset of the activity of the New Orleans Facebook  community over three months
    ~\cite{viswanath2009evolution}. The Tumblr graph contains quoting between Tumblr users and is a subset of the Memetracker\footnote{\url{snap.stanford.edu/data/memetracker9.html}} data set.  Rozenshtein et al.~\cite{rozenshtein2016reconstructing} used these  graphs and epidemic simulations to reconstruct dissemination processes.
	\item \emph{DBLP:} We used a subset of the DBLP\footnote{\url{https://dblp.uni-trier.de/}} database to generate temporal co-author graphs. The subset was chosen by considering publications in proceedings of selected machine learning conferences. 
	The time stamp of an edge is the year of the joint publication. 
\end{dsitemize}
To obtain data sets for supervised graph classification, we generated induced subgraphs by starting a BFS run from each vertex. 
We terminated the BFS when at most $20$ vertices in case of the MIT, $50$ vertices in case of the Infectious, $60$ vertices in case of the Highschool, Tumblr or DBLP, and $100$ vertices in case of the Facebook graph had been added.
Using the above procedure, we generated between $97$ and $995$ graphs for each of the data sets.
See \cref{table:datasets_stats2} %
 for data set statistics. All data sets will be made publicly available. %
In the following we describe the model for the dissemination process and the classification tasks.
\\\\
\emph{Dissemination Process Simulation}~~
We simulated a dissemination process on each of the induced subgraphs according to the \emph{susceptible-infected (SI)} model---a common approach to simulate epidemic spreading, e.g., see~\cite{bai2017optimizing}. %
In the SI model, each vertex is either \new{susceptible} or \new{infected}. 
An infected vertex is part of the temporal dissemination process. %
An initial seed of $s$ vertices is selected randomly and labeled as infected. %
Infections propagate in discrete time steps along temporal edges with a fixed probability $0<p \leq 1$. If a vertex is infected it stays infected indefinitely. 
A newly infected vertex may infect its neighbors at the next time step.
The simulation runs until at least $|V|\cdot I$ vertices with $0<I\leq 1$ are infected, or no more infections are possible.  
\\\\
\emph{Classification Tasks}~~
We consider two classification tasks.  
The first is the discrimination of temporal graphs with vertex labels corresponding to observations of a dissemination process and temporal graphs in which the labeling is not a result of a dissemination process.
Here, for each data set, we run the SI simulation with equal parameters of $s=1$, $p=0.5$ and $I=0.5$  for all graphs.
We used half of the data set as our first class. For our second class, we used the remaining graphs. For each graph in the second class, we counted the number $V_{\text{inf}}$ of infected vertices, reset the labels of all vertices back to uninfected, and finally infected $V_{\text{inf}}$ vertices randomly at a random time.

The second classification task is the discrimination of temporal graphs that differ in the dissemination processes itself.
Therefore, we run the SI simulation with different parameters for each of the two subsets.
For both subsets we set $s=1$ and $I=0.5$.
However, for the first subset of graphs we set the infection probability $p=0.2$ and for the second subset we set $p=0.8$.
The simulation runs repeatedly until at least $|V|\cdot I$ vertices are infected or no more infections are possible.
Notice that a classification by only counting the number of infected vertices is impossible for the classification tasks.

In order to evaluate our methods under conditions with incomplete information, we generated additional data sets based on \emph{Infectious} for both classification tasks.
For each graph, we randomly set the labels of $\{10\%,\ldots,80\%\}$ of the infected vertices back to non-infected.
We repeated this ten times resulting in 80 data sets for each of the two classification tasks.

\begin{table*}\centering	%
\caption{Running times in ms for the first classification task, random walk length $k=3$ ($k=2$ for \emph{DL-RW}) and number of iterations of WL $h=3$ ($h=2$ for \emph{DL-WL}).} %
\label{table:runningtimes1}
\resizebox{0.75\textwidth}{!}{ 	\renewcommand{\arraystretch}{.8}
    \begin{tabular}{@{}c <{\enspace}@{}l@{\hspace{1mm}}ccccccc@{}}	\toprule
        
        & \multirow{3}{*}{\vspace*{4pt}\textbf{Kernel}}&\multicolumn{6}{c}{\textbf{Data set}}\\\cmidrule{3-8}
        &  & {\textsc{Mit}}        & \textsc{Highschool}  & \textsc{Infectious}&  \textsc{Tumblr}   & \textsc{Dblp}  & \textsc{Facebook} \\ \toprule
        
        \multirow{2}{*}{\rotatebox{90}{\small Stat.}}
        & \emph{Stat-RW} & $50330$ & $124060$  & $130630$ & $29209$ & $1090141$ & $456421$ \\ 
        & \emph{Stat-WL} & $87$ & $106$  & $129$ & $164$ & $410$ & $\textbf{715}$ \\ 
        \cmidrule{1-8}\multirow{9}{*}{\rotatebox{90}{\small Temporal}}
        & \emph{RD-RW} & $270$ & $4017$  & $10259$ & $14563$ & $99327$ & $357414$ \\ 
        & \emph{RD-WL} & $\textbf{22}$ & $\textbf{84}$  & $\textbf{96}$ & $\textbf{145}$ & $\textbf{405}$ & $757$ \\ \addlinespace[0.8mm] 
        & \emph{DL-RW} & $26447614$ & $33643$  & $10764$ & $2887$ & $6715$ & $5050$ \\ 
        & \emph{DL-WL} & $71636$ & $1972$  & $847$ & $429$ & $1242$ & $1131$ \\ \addlinespace[0.8mm] 
        & \emph{SE-RW} & $7211$ & $3660$  & $4973$ & $2018$ & $6454$ & $5173$ \\ 
        & \emph{SE-WL} & $581$ & $382$  & $262$ & $302$ & $1097$ & $1359$ \\ \addlinespace[0.8mm] 
        & \emph{APPROX (S=50)} & $188$ & $301$  & $309$ & $409$ & $1261$ & $1926$ \\ 
        & \emph{APPROX (S=100)} & $363$ & $412$  & $599$ & $944$ & $1982$ & $3452$ \\ 
        & \emph{APPROX (S=250)} & $896$ & $1589$  & $1451$ & $1901$ & $5044$ & $7607$ \\ 
        \bottomrule  \end{tabular}}
\end{table*}
\begin{table*}[t]\centering	
    \caption{Running times in ms for the second classification task, random walk length $k=3$ ($k=2$ for \emph{DL-RW}) and number of iterations of WL $h=3$ ($h=2$ for \emph{DL-WL}).} 
    \label{table:runningtimes2}
    \resizebox{0.75\textwidth}{!}{ 	\renewcommand{\arraystretch}{.8}
        \begin{tabular}{@{}c <{\enspace}@{}l@{\hspace{3mm}}ccccccc@{}}	\toprule
            
            & \multirow{3}{*}{\vspace*{4pt}\textbf{Kernel}}&\multicolumn{6}{c}{\textbf{Data set}}\\\cmidrule{3-8}
            &  & {\textsc{Mit}}        & \textsc{Highschool}  & \textsc{Infectious}&  \textsc{Tumblr}   & \textsc{Dblp}  & \textsc{Facebook}  \\ \toprule
            
            \multirow{2}{*}{\rotatebox{90}{Stat.}}
            & \emph{Stat-RW} & $48497$ & $107947$  & $139620$ & $32353$ & $1266452$ & $529073$ \\ 
            & \emph{Stat-WL} & $73$ & $106$  & $128$ & $164$ & $401$ & $\textbf{668}$ \\ 
            \cmidrule{1-8}\multirow{9}{*}{\rotatebox{90}{Temporal}}
            & \emph{RD-RW} & $294$ & $4516$  & $18432$ & $18436$ & $114842$ & $413073$ \\ 
            & \emph{RD-WL} & $\textbf{25}$ & $\textbf{83}$  & $\textbf{98}$ & $\textbf{141}$ & $\textbf{385}$ & $729$ \\ \addlinespace[0.8mm] 
            & \emph{DL-RW} & $26274726$ & $32919$  & $10590$ & $2890$ & $9548$ & $4380$ \\ 
            & \emph{DL-WL} & $71180$ & $2121$  & $806$ & $423$ & $1196$ & $1054$ \\ \addlinespace[0.8mm] 
            & \emph{SE-RW} & $7462$ & $3743$  & $4820$ & $1792$ & $6173$ & $6468$ \\ 
            & \emph{SE-WL} & $575$ & $369$  & $261$ & $293$ & $846$ & $1126$ \\ \addlinespace[0.8mm] 
            & \emph{APPROX (S=50)} & $185$ & $360$  & $305$ & $476$ & $1170$ & $1590$ \\ 
            & \emph{APPROX (S=100)} & $361$ & $573$  & $595$ & $807$ & $2318$ & $3696$ \\ 
            & \emph{APPROX (S=250)} & $897$ & $1706$  & $1525$ & $1869$ & $5093$ & $7318$ \\ 
            \bottomrule  \end{tabular}}
\end{table*}
\begin{figure*}[t]
    \centering
    \begin{subfigure}{0.44\textwidth}
        \centering	
        \includegraphics[width=1\linewidth]{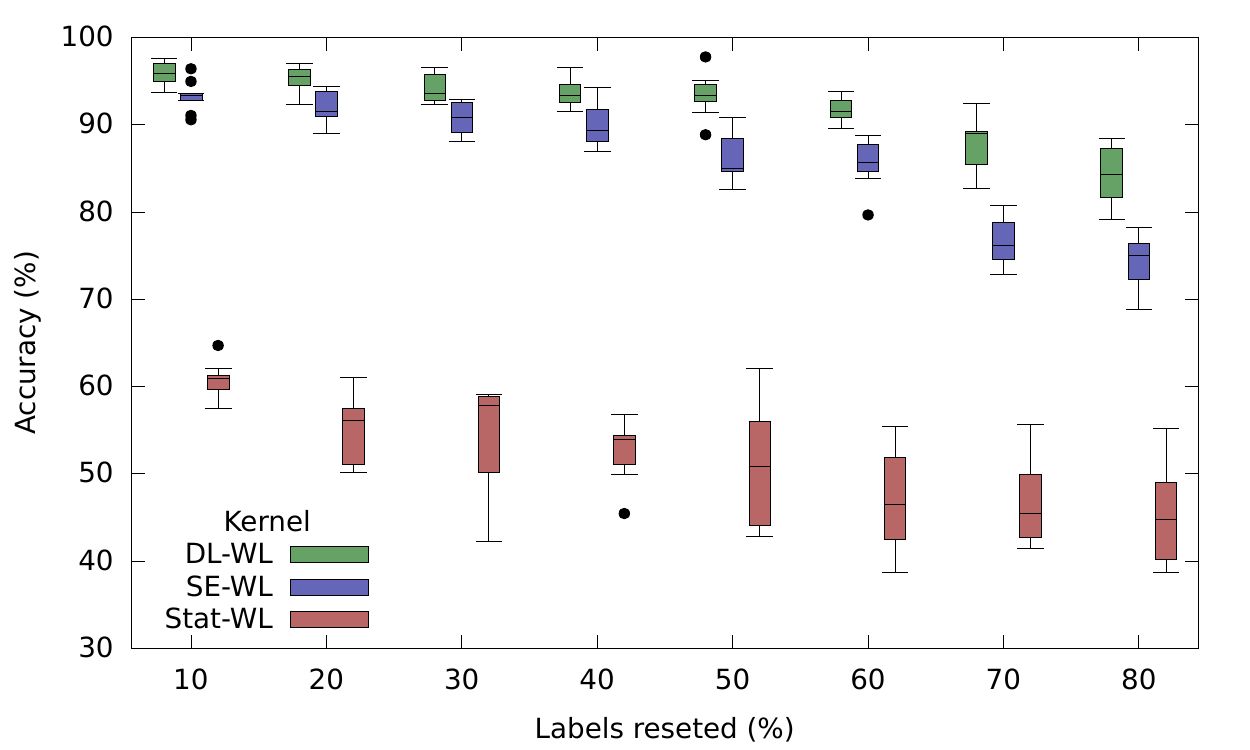}
        \subcaption{Results for the first classification task.}
        \label{fig:dataa}        
    \end{subfigure}
    \begin{subfigure}{0.44\textwidth}
        \centering
        \includegraphics[width=1\linewidth]{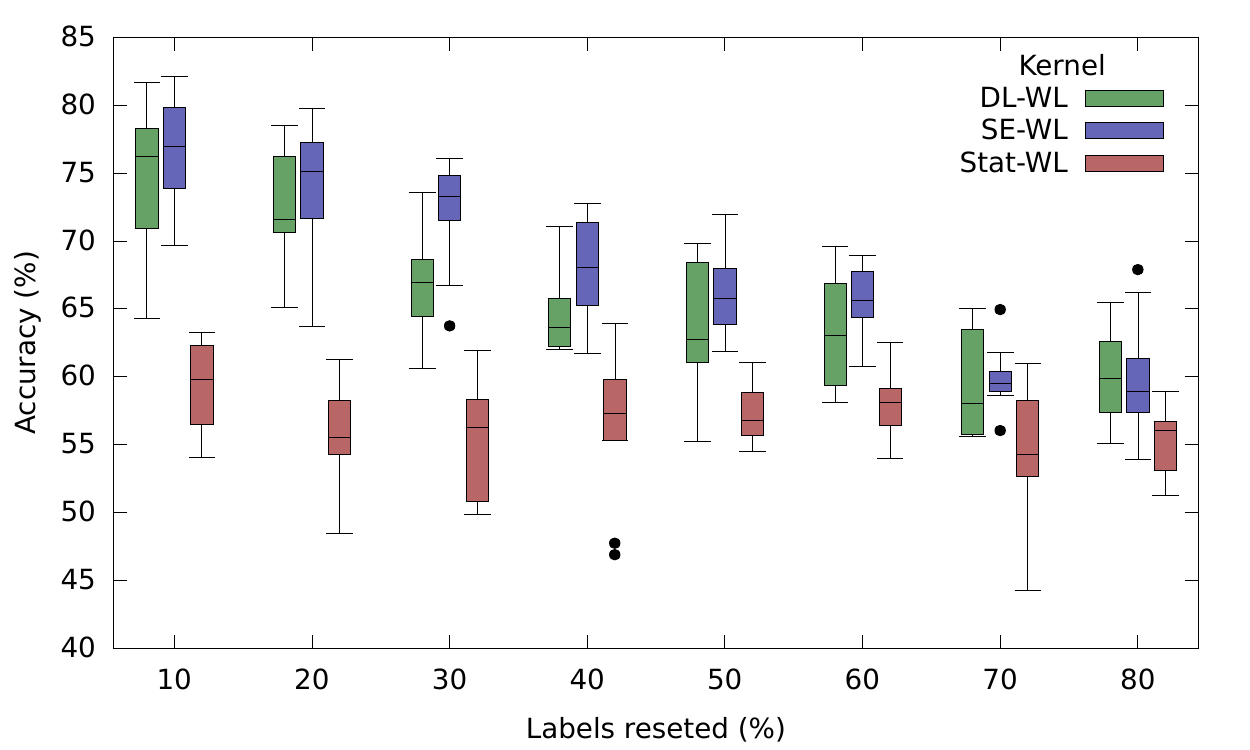}
        \subcaption{Results for the second classification task.}
        \label{fig:datab}
    \end{subfigure}
    \caption{Results of  \emph{DL-WL}, \emph{SE-WL} and \emph{Stat-WL} kernel for the \emph{Infectious} data set under incomplete data. }
    \label{fig:uncdata}
    \vspace{-10pt}
\end{figure*}
\subsection{Graph kernels} 
As a baseline we use the $k$-step random walk (\emph{Stat-RW}) and the Weisfeiler-Lehman subtree (\emph{Stat-WL}) kernel on the static graphs obtained by interpreting the time stamps as discrete edge labels, and assigning to each vertex the concatenated sequence of its labels.
To evaluate the three approaches of \cref{sec:temporalk}, we use the $k$-step random walk and the Weisfeiler-Lehman subtree kernel, resulting in the following kernel instances: (1) \emph{RD-RW} and \emph{RD-WL}, which use the reduced graph representation (\Cref{subsec:rdrep}), (2) \emph{DL-RW} and \emph{DL-WL}, which use the directed line graph expansion (\Cref{subsec:dlk}), (3)  \emph{SE-RW} and \emph{SE-WL}, which use the static expansion (\Cref{subsec:sek}).
We evaluate the approximation (\emph{APPROX}) for the directed line graph expansion, proposed in \Cref{sec:approx}, with sample sizes $S=50$, $S=100$ and $S=250$. 
\subsection{Experimental Protocol}\label{ep}

For each kernel, we computed the normalized Gram matrix. We report the classification accuracies obtained with the $C$-SVM implementation of \text{LIBSVM}~\cite{Cha+2011}, using 10-fold cross validation. The $C$-parameter was selected from $\{10^{-3}, 10^{-2}, \dotsc, 10^{2},$ $10^{3}\}$ by 10-fold cross validation on the training folds. 			
We repeated each 10-fold cross validation ten times with different random folds, and report average accuracies and standard deviations.
The number of steps of the random walk kernel ($k\in\{0,\ldots,5\}$) and the number of iterations of the  Weisfeiler-Lehman subtree kernel ($h\in\{0,\ldots,5\}$) were selected by fold-wise 10-fold cross-validation. 
All experiments were conducted on a workstation with an 
\text{Intel Xeon E5-2640v3} with 2.60\si GHz and 128\si GB of RAM running \text{Ubuntu 16.04.6} LTS using a single core. We used GNU \CC Compiler 5.5.0 with the flag \texttt{--O2}.\footnote{The code will available at \url{https://www.github.com}}
To compare running times, we set the walk length of \emph{DL-RW} to $k=2$, and for \emph{DL-WL} we set the number of iterations to $h=2$.

\subsection{Results and Discussion}
In the following we answer questions \textbf{Q1} to \textbf{Q3}.  
\\\\
\textbf{Q1}	
\Cref{table:results} and \Cref{table:results2} show that taking temporal information into account is crucial.
Our approaches lead to improvements in accuracy over all data sets. In most cases the improvement is substantial. 
For the first classification task, \emph{DL-RW} and \emph{DL-WL} reach the best accuracies for all but the \textsc{Tumblr} data set, here \emph{SE-RW} is best.
However, also for the other data sets \emph{SE-RW} and \emph{SE-WL} are on par with slightly lower accuracies.
For the second classification task, we have a similar situation, our approaches beat the static kernels in all cases. The \emph{Stat-RW} and \emph{Stat-WL} kernels have a significantly lower accuracy for all data sets and are not able to successfully detect dissemination processes.
This classification task poses a greater challenge for the temporal kernels which reach less good results compared to the first classification task.
Especially the \textsc{Mit} data set seems to be hard, only the \emph{DL-RW} reaches an accuracy of over $80\%$.
However, it also has the overall highest running time for this data set due to its quadratic blowup. See \Cref{table:runningtimes1} for the running times of the first classification task (\Cref{table:runningtimes2} shows similar values for the second task). %
The running times for the random walk kernels are by orders of magnitude higher than the ones of the Weisfeiler-Lehman kernels.
The reduced graph kernels cannot compete  with our other approaches in terms of accuracy. In particular for the second classification task the loss of temporal information led to lower accuracies.
However, the running times, especially of \emph{RD-WL} are low. 
For a lower average number of temporal edges and vertex degree, its advantage gained by reducing the number of edges decreases, and with larger data sets the running times increase. \emph{RD-RW} and \emph{RD-WL} deliver slightly worse results for the \emph{Facebook} data set compared to the static kernels for both tasks.
\\
\textbf{Q2} 
For a sample size of $S=50$ \emph{APPROX} performs better than the static kernels. %
And, the accuracies are on par or better than the ones of the reduced graph kernels.
With a larger sample sizes of $S=100$ and $S=250$ the gap between the accuracies of \emph{APPROX} and \emph{DL-RW} is reduced for all data sets in both classification tasks. 
\Cref{table:runningtimes1} shows that the running time of the approximation algorithm is by orders of magnitude faster for the \textsc{Mit} data set.
For $S=50$ and $S=100$ there is an improvement in running times for all data sets. For $S=250$ the running times of the exact algorithm for the \emph{Facebook} data set is faster. 
\\
\textbf{Q3}
We ran the Weisfeiler-Lehman subtree kernels for the \emph{Infectious} data sets where formerly infected vertices were randomly set to non-infected. 
For the first classification task \emph{DL-WL} and \emph{SE-WL} keep high average accuracy, see \cref{fig:dataa}. 
The \emph{Stat-WL} kernel falls under $50\%$ accuracy. 
For the second task the \emph{SE-WL} kernel achieves better average accuracy than the \emph{DL-WL} kernel for up to $70\%$ of reset labels, see \Cref{fig:datab}.
Only for $80\%$ the \emph{DL-WL} kernel achieves better average accuracy.
\section{Conclusion}\label{sec:conclusion}
We introduced a framework lifting static kernels to the temporal domain, and obtained variants of the Weisfeiler-Lehman subtree and the $k$-step random walk kernel. 
Furthermore, we introduced a stochastic kernel directly based on temporal walks with provable approximation guarantees.
We empirically evaluated our methods on real-world social networks showing that incorporating temporal information is crucial for classifying temporal graphs under consideration of dissemination processes. 
Moreover, we showed that the approximation approach performs well and is able to speed up computation by orders of magnitude. %
Additionally, we demonstrated that our proposed kernels work in scenarios where information of the dissemination process is incomplete or missing.
We believe that our techniques are a stepping stone for developing neural approaches for temporal graph representation learning. 

\balance
\bibliographystyle{plain}
\bibliography{literature}

\begin{thebibliography}{10}

\bibitem{Adams2016}
N.~Adams and N.~Heard.
\newblock {\em Dynamic Networks and Cyber-Security}.
\newblock Imperial College Press, 2016.

\bibitem{Anil2014}
A.~Anil, N.~Sett, and S.~R. Singh.
\newblock Modeling evolution of a social network using temporal graph kernels.
\newblock In {\em ACM SIGIR}, pages 1051--1054, 2014.

\bibitem{bai2017optimizing}
Y.~Bai, B.~Yang, L.~Lin, J.~L. Herrera, Z.~Du, and P.~Holme.
\newblock Optimizing sentinel surveillance in temporal network epidemiology.
\newblock {\em Scientific reports}, 7(1):4804, 2017.

\bibitem{Bor+2005}
K.~M. Borgwardt and H.-P. Kriegel.
\newblock Shortest-path kernels on graphs.
\newblock In {\em IEEE ICDM}, pages 74--81, 2005.

\bibitem{Cha+2011}
C.-C. Chang and C.-J. Lin.
\newblock {LIBSVM}: {A} library for support vector machines.
\newblock {\em ACM Transactions on Intelligent Systems and Technology},
  2:27:1--27:27, 2011.

\bibitem{konect:eagle06}
N.~Eagle and A.~Pentland.
\newblock {Reality} {Mining}: Sensing complex social systems.
\newblock {\em Personal Ubiquitous Computing}, 10(4):255--268, 2006.

\bibitem{Gaertner2003}
T.~G\"{a}rtner, P.~Flach, and S.~Wrobel.
\newblock On graph kernels: {H}ardness results and efficient alternatives.
\newblock In {\em Learning Theory and Kernel Machines}, pages 129--143. 2003.

\bibitem{Gil+2017}
J.~Gilmer, S.~S. Schoenholz, P.~F. Riley, O.~Vinyals, and G.~E. Dahl.
\newblock Neural message passing for quantum chemistry.
\newblock In {\em ICML}, pages 1263--1272, 2017.

\bibitem{grindrod2011communicability}
P.~Grindrod, M.~C. Parsons, D.~J. Higham, and E.~Estrada.
\newblock Communicability across evolving networks.
\newblock {\em Physical Review E}, 83(4):046120, 2011.

\bibitem{Harary1960}
F.~Harary and R.~Z. Norman.
\newblock Some properties of line digraphs.
\newblock {\em Rendiconti del Circolo Matematico di Palermo}, 9(2):161--168,
  May 1960.

\bibitem{holme2013epidemiologically}
P.~Holme.
\newblock Epidemiologically optimal static networks from temporal network data.
\newblock {\em PLoS computational biology}, 9(7):e1003142, 2013.

\bibitem{holme2015modern}
P.~Holme.
\newblock Modern temporal network theory: a colloquium.
\newblock {\em The European Physical Journal B}, 88(9):234, 2015.

\bibitem{Isella2011}
L.~Isella, J.~Stehlé, A.~Barrat, C.~Cattuto, J.-F. Pinton, and W.~Van~den
  Broeck.
\newblock What's in a crowd? {A}nalysis of face-to-face behavioral networks.
\newblock {\em Journal of Theoretical Biology}, 271(1):166--180, 2011.

\bibitem{Knauf2016}
K.~Knauf, D.~Memmert, and U.~Brefeld.
\newblock Spatio-temporal convolution kernels.
\newblock {\em Machine Learning}, 102(2):247--273, Feb 2016.

\bibitem{Kon+2016}
R.~Kondor and H.~Pan.
\newblock The multiscale {L}aplacian graph kernel.
\newblock In {\em NIPS}, pages 2982--2990, 2016.

\bibitem{Kri+2016}
N.~M. Kriege, P.-L. Giscard, and R.~C. Wilson.
\newblock On valid optimal assignment kernels and applications to graph
  classification.
\newblock In {\em NIPS}, pages 1615--1623, 2016.

\bibitem{Kriege2019}
N.~M. Kriege, F.~D. Johansson, and C.~Morris.
\newblock A survey on graph kernels.
\newblock {\em CoRR}, abs/1903.11835, 2019.

\bibitem{Kri+2017b}
N.~M. Kriege, M.~Neumann, C.~Morris, K.~Kersting, and P.~Mutzel.
\newblock A unifying view of explicit and implicit feature maps for structured
  data: Systematic studies of graph kernels.
\newblock {\em CoRR}, abs/1703.00676, 2017.

\bibitem{leskovec2007cost}
J.~Leskovec, A.~Krause, C.~Guestrin, C.~Faloutsos, C.~Faloutsos, J.~VanBriesen,
  and N.~Glance.
\newblock Cost-effective outbreak detection in networks.
\newblock In {\em ACM KDD 2007}, pages 420--429. ACM, 2007.

\bibitem{Li+2015}
L.~Li, H.~Tong, Y.~Xiao, and W.~Fan.
\newblock \emph{Cheetah}: Fast graph kernel tracking on dynamic graphs.
\newblock In {\em SDM}, pages 280--288, 2015.

\bibitem{michail2016introduction}
O.~Michail.
\newblock An introduction to temporal graphs: An algorithmic perspective.
\newblock {\em Internet Mathematics}, 12(4):239--280, 2016.

\bibitem{michail2016traveling}
O.~Michail and P.~G. Spirakis.
\newblock Traveling salesman problems in temporal graphs.
\newblock {\em Theoretical Computer Science}, 634:1--23, 2016.

\bibitem{Mor+2017}
C.~Morris, K.~Kersting, and P.~Mutzel.
\newblock Glocalized {W}eisfeiler-{L}ehman kernels: Global-local feature maps
  of graphs.
\newblock In {\em IEEE ICDM}, pages 327--336, 2017.

\bibitem{Mor+2019}
C.~Morris, M.~Ritzert, M.~Fey, W.~L. Hamilton, J.~E. Lenssen, G.~Rattan, and
  M.~Grohe.
\newblock Weisfeiler and {L}eman go neural: Higher-order graph neural networks.
\newblock In {\em AAAI Conference on Artificial Intelligence}, volume~33, pages
  4602--4609, 2019.

\bibitem{Nguyen2018}
G.~H. Nguyen, J.~B. Lee, R.~A. Rossi, N.~K. Ahmed, E.~Koh, and S.~Kim.
\newblock Continuous-time dynamic network embeddings.
\newblock In {\em The Web Conference}, pages 969--976, 2018.

\bibitem{Nik+2018}
G.~Nikolentzos, P.~Meladianos, S.~Limnios, and M.~Vazirgiannis.
\newblock A degeneracy framework for graph similarity.
\newblock In {\em IJCAI}, pages 2595--2601, 2018.

\bibitem{Nik+2017}
G.~Nikolentzos, P.~Meladianos, and M.~Vazirgiannis.
\newblock Matching node embeddings for graph similarity.
\newblock In {\em AAAI}, pages 2429--2435, 2017.

\bibitem{oettershagen2020temporal}
Lutz Oettershagen, Nils~M Kriege, Christopher Morris, and Petra Mutzel.
\newblock Temporal graph kernels for classifying dissemination processes.
\newblock In {\em Proceedings of the 2020 SIAM International Conference on Data
  Mining}, pages 496--504. SIAM, 2020.

\bibitem{Paassen/etal/2017a}
B.~Paa{\ss}en, C.~G{\"o}pfert, and B.~Hammer.
\newblock Time series prediction for graphs in kernel and dissimilarity spaces.
\newblock {\em Neural Processing Letters}, pages 1--21, 2017.

\bibitem{rozenshtein2016reconstructing}
P.~Rozenshtein, A.~Gionis, B.~A. Prakash, and J.~Vreeken.
\newblock Reconstructing an epidemic over time.
\newblock In {\em ACM KDD}, pages 1835--1844. ACM, 2016.

\bibitem{She+2011}
N.~Shervashidze, P.~Schweitzer, E.~J. van Leeuwen, K.~Mehlhorn, and K.~M.
  Borgwardt.
\newblock Weisfeiler-{L}ehman graph kernels.
\newblock {\em Journal of Machine Learning Research}, 12:2539--2561, 2011.

\bibitem{Sugiyama2015}
M.~Sugiyama and K.~M. Borgwardt.
\newblock Halting in random walk kernels.
\newblock In {\em NIPS}, pages 1639--1647, 2015.

\bibitem{Tri+2019}
R.~Trivedi, M.~Farajtabar, P.~Biswal, and H.~Zha.
\newblock Dyrep: Learning representations over dynamic graphs.
\newblock In {\em ICLR}, 2019.

\bibitem{viswanath2009evolution}
B.~Viswanath, A.~Mislove, M.~Cha, and K.~P. Gummadi.
\newblock On the evolution of user interaction in facebook.
\newblock In {\em ACM Workshop on Online Social Networks}, pages 37--42, 2009.

\bibitem{Vos+2018}
S.~Vosoughi, D.~Roy, and S.~Aral.
\newblock The spread of true and false news online.
\newblock {\em Science}, 359(6380):1146--1151, 2018.

\bibitem{wang2018time}
H.~Wang, J.~Wu, X.~Zhu, Y.~Chen, and C.~Zhang.
\newblock Time-variant graph classification.
\newblock {\em IEEE Transactions on Systems, Man, and Cybernetics: Systems},
  2018.

\bibitem{who2019rep43}
{World Health Organization (WHO)}.
\newblock Ebola virus disease, democratic republic of the congo, external
  situation report 43.
\newblock 2019.

\bibitem{Wu/etal/2014a}
B.~Wu, C.~Yuan, and W.~Hu.
\newblock Human action recognition based on context-dependent graph kernels.
\newblock In {\em IEEE CVPR}, pages 2609--2616, 2014.

\end{thebibliography}
\end{document}